\documentclass[10pt,a4paper]{article}
\usepackage[utf8]{inputenc}
\usepackage[english]{babel}
\usepackage{amsmath}
\usepackage{amsfonts}
\usepackage{amssymb}
\usepackage{amsthm}
\usepackage{cite}
\usepackage{graphicx}
\usepackage[left=2cm,right=2cm,top=2cm,bottom=2cm]{geometry}
\newtheorem{te}{Theorem}

\newtheorem{lem}{Lemma}

\newtheorem{de}{Definition}

\providecommand{\keywords}[1]
{\small	\textbf{\textit{Keywords:}} #1}
\title{Lie integrability by quadratures for symplectic, cosymplectic, contact and cocontact Hamiltonian systems}
\author{R. Azuaje \\ Departamento de Física, Universidad Autónoma Metropolitana Unidad Iztapalapa\\ San Rafael Atlixco 186, C.P. 09340, Ciudad de México, México \\ e-mail: razuaje@xanum.uam.mx}
\begin{document}
\maketitle

\begin{abstract}
In this paper we present the theorem on Lie integrability by quadratures for time-independent Hamiltonian systems on symplectic and contact manifolds, and for time-dependent Hamiltonian systems on cosymplectic and cocontact manifolds. We show that having a solvable Lie algebra of constants of motion for a Hamiltonian system is equivalent to having a solvable Lie algebra of symmetries of the vector field defining the dynamics of the system, which allows us to find solutions of the equations of motion by quadratures.

\keywords{Solvable Lie algebras. Symmetries. Contact Hamiltonian systems. Quadratures.}

\end{abstract}

\section{Introduction}
\label{sec1}

It is well known that the existence of symmetries for a mechanical system, allows a reduction of the system or in other words allows a simplification of the problem of finding the equations of motion. It is also well known, at least in the symplectic framework, that constants of motion are related to symmetries \cite{AMRC2019,AKN2006,Roman2020}. The famous Liouville theorem for time-independent and time dependent Hamiltonian systems on symplectic and cosymplectic manifolds respectively, asserts that if for a Hamiltonian system with $n$ degrees of freedom (the dimension of the phase space is $2n$ or $2n+1$), the knowledge of $n$ constants of motion (or conserved quantities) in involution allows us to find the solutions of the equations of motion by quadratures, i.e. by using finitely many algebraic operations (including taking inverse functions) and calculating the integrals of known functions \cite{AMRC2019,AKN2006,GPS2002,BBT2003,LR89,TC2018,BB98,GVY2008}. So we say that a Hamiltonian system with $n$ degrees of freedom is integrable in the sense of Liouville if we can find $n$ independent constants of motion in involution. In involution means that the constants of motion generate an Abelian Lie algebra of functions with the Poisson bracket, or equivalently, the Hamiltonian vector fields related to such constants of motion are symmetries of the dynamical vector field that generate an Abelian Lie algebra of vector fields with the Lie bracket (or commutator), i.e. the symmetries commute to each other. For a dynamical system, the notion of Lie integrability by quadratures is based on the existence of a solvable Lie algebra of symmetries of the dynamical vector field \cite{AKN2006,Kozlov96,CFGR2015,GG2020,Cetal2016}, which is a more general notion than Liouville integrability. The concept of solvable Lie algebra is more general than the one of Abelian Lie algebra, indeed, every Abelian Lie algebra is trivially a solvable Lie algebra but the converse statement is not true.

The concept of a Lie algebra is due to the Norwegian mathematician Marius
Sophus Lie (1842–1899), who developed the theory of Lie groups and Lie algebras,
trying to solve, or at least simplify, ordinary differential equations. At the beginning,
Lie based his work on analogies with the group theory developed by Evariste Galois
(1811–1832) in order to solve algebraic equations that were quadratic, cubic or
quartic \cite{Gilmor2008}. Lie’s theorem establishes that if for a vector field we have sufficient linearly independent symmetries that generate a solvable Lie algebra over the Lie bracket of vector fields, then the solutions of the equations of motion of the dynamical system defined by such vector field can be found by quadratures.

The aim of this paper is to present under the formalisms of symplectic geometry, cosymplectic geometry, contact geometry and cocontact geometry the theorem of Lie integrability by quadratures for Hamiltonian systems. Contact and cocontact Hamiltonian mechanics is a subject of active research these days, see for example \cite{LL2020,BG2021,GLR2022,AE2023}. We follow the spirit of Lie, i.e., we focus on finding solutions of the Hamilton equations of motion by quadratures, in this paper we do not address the problem of the existence of action-angle coordinates. This paper is organized as follows: In section \ref{sec2} we present a brief review of Hamiltonian systems on symplectic, cosymplectic, contact and cocontact manifolds; in section \ref{sec3} we present a brief review of the Lie's theorem on integrability by quadratures of a dynamical system defined by a smooth vector field on $\mathbb{R}^{n}$. In section \ref{sec4} we present  the Lie's theorem on integrability by quadratures for (time-independent) symplectic and (time-dependent) cosymplectic Hamiltonian systems. Finally in section \ref{sec5} we present  the Lie's theorem on integrability by quadratures for time-independent and time-dependent contact Hamiltonian systems.

To finish this introduction, we present the concept of Lie algebra and solvable Lie algebra in order to set up the language and notation employed in this paper, for details see \cite{Gilmor2008}.

\begin{de}
A Lie algebra over $\mathbb{R}$ is a real vector space $\mathfrak{g}$ equipped with a bilinear map  $[,]:\mathfrak{g}\times \mathfrak{g} \longrightarrow \mathfrak{g}$ called Lie bracket that satisfies $[X,Y]=-[Y,X]$ $\forall X,Y\in \mathfrak{g}$ and
$[X,[Y,Z]]+[Y,[Z,X]]+[Z,[X,Y]]=0$ $\forall X,Y,Z\in \mathfrak{g}$.
\end{de}

\begin{de}
Let $\mathfrak{g}$ be a Lie algebra.
\begin{itemize}
\item A Lie subalgebra of $\mathfrak{g}$ is a vector subspace $L$ of $\mathfrak{g}$ such that $[X,Y]\in L$ $\forall X,Y\in L$.
\item A vector subspace $I$ of $\mathfrak{g}$ is called an ideal of $\mathfrak{g}$ if $[X,Y]\in I$ $\forall X\in I$ and $\forall Y\in \mathfrak{g}$. Note that an ideal is also a Lie subalgebra.
\item We say that $\mathfrak{g}$ is a solvable Lie algebra if there exists a finite sequence of Lie subalgebras of $\mathfrak{g}$, let us say $L_{1},\ldots,L_{n}$ with $n=dim(\mathfrak{g})$, such that $\lbrace 0\rbrace =L_{0}\subset L_{1} \subset \cdots \subset L_{n-1}\subset L_{n}=\mathfrak{g}$ and $L_{i}$ is an ideal of $L_{i+1}$ of codimension $1$ for $i=0,1,\ldots,n-1$.
\end{itemize}
\end{de}

\section{Symplectic, cosymplectic, contact and cocontact Hamiltonian systems}
\label{sec2}

In this section we review the formalism of symplectic geometry, cosymplectic geometry, contact geometry, cocontact geometry and the formulations of Hamiltonian mechanics under these geometries. The purpose of this section is to set up the notation and terminology used in this paper.

Let us start with a brief review of symplectic geometry and the formulation of time-independent Hamiltonian mechanics (for details see \cite{AMRC2019,LR89,Torres2020,Lee2012}).
 
Let $(M,\omega)$ be a symplectic manifold of dimension $2n$. Around any point $p\in M$ there exist local coordinates $(q^{1},\cdots,q^{n},p_{1},\cdots,p_{n})$, called canonical coordinates or Darboux coordinates, such that
\begin{equation}
\omega=dq^{i}\wedge dp_{i}.
\end{equation}
In this paper we adopt the Einstein summation convention ( i.e., a summation over repeated indices is assumed).

For each $f\in C^{\infty}(M)$ is assigned a vector field $X_{f}$ on $M$, called the Hamiltonian vector field for $f$, according to
\begin{equation}
X_{f}\lrcorner \omega \ = \ df \ .
\end{equation}
In canonical coordinates, $X_{f}$ reads
\begin{equation}
X_{f}=\frac{\partial f}{\partial p_{i}}\frac{\partial}{\partial q^{i}}-\frac{\partial f}{\partial q^{i}}\frac{\partial}{\partial p_{i}}\ .
\end{equation}
The assignment $f\longmapsto X_{f}$ is linear, that is
\begin{equation}
X_{f+\alpha g}\ = \ X_{f}+\alpha X_{g}\ ,
\end{equation}
$\forall f,g\in C^{\infty}(M)$ and $\forall \alpha \in\mathbb{R}$. Given $f,g \in C^{\infty}(M)$ the Poisson bracket of $f$ and $g$ is defined by
\begin{equation}
\lbrace f,g\rbrace=X_{g}f=\omega(X_{f},X_{g}).
\end{equation}
In canonical coordinates we have
\begin{equation}
\lbrace f,g\rbrace \ = \ \frac{\partial f}{\partial q^{i}}\frac{\partial g}{\partial p_{i}}\,-\,\frac{\partial f}{\partial p_{i}}\frac{\partial g}{\partial q^{i}} \ .
\end{equation}

The theory of time-independent Hamiltonian systems is naturally constructed within the mathematical formalism of symplectic geometry. Given $H\in C^{\infty}(M)$, the dynamics of the Hamiltonian system on $(M,\omega)$ (the phase space) with Hamiltonian function $H$ is defined by the Hamiltonian vector field $X_{H}$, that is, the trajectories of the system $\psi(t)=(q^{1}(t),\cdots,q^{n}(t),p_{1}(t),\cdots,p_{n}(t))$ are the integral curves of $X_{H}$, they satisfy the Hamilton's equations of motion
\begin{equation}
\dot{q^{i}} =\frac{\partial H}{\partial p_{i}}, \hspace{1cm}
\dot{p_{i}} =-\frac{\partial H}{\partial q^{i}}\qquad ;\qquad i=1,2,3,\ldots,n \ .
\end{equation}

The evolution (the temporal evolution) of a function $f\in C^{\infty}(M)$ (a physical observable) along the trajectories of the system is given by
\begin{equation}
\dot{f}= L_{X_{H}}f=X_{H}f=\lbrace f,H\rbrace \ ,
\end{equation}
where $L_{X_{H}}f$ is the Lie derivative of $f$ with respect to $X_{H}$. We say that $f$ is a constant of motion of the system if it is constant along the trajectories of the system, that is, $f$ is a constant of motion if $L_{X_{H}}f=0$ ($\lbrace f,H\rbrace=0$).

Now let us continue with cosymplectic geometry and the formulation of time-dependent Hamiltonian systems under cosymplectic manifolds. 

\begin{de}
Let $M$ be a $2n+1$ dimensional smooth manifold. A cosymplectic structure on $M$ is a couple $(\Omega,\eta)$, where $\Omega$ is a closed 2-form on $M$ and $\eta$ is a closed 1-form on $M$ such that $\eta\wedge\Omega^{n}\neq 0$. If $(\Omega,\eta)$ is a cosymplectic structure on $M$ we say that $(M,\Omega,\eta)$ is a cosymplectic manifold.
\end{de}

Let $(M,\Omega,\eta)$ be a cosymplectic manifold of dimension $2n+1$. Around any point $p\in M$ there exist local coordinates $(q^{1},\cdots,q^{n},p_{1},\cdots,p_{n},t)$, called canonical coordinates or Darboux coordinates, such that
\begin{equation}
\Omega=dq^{i}\wedge dp_{i}\hspace{1cm}\mathrm{and}\hspace{1cm}\eta=dt.
\end{equation}
There exists a distinguished vector field $R$ on $M$, called the Reeb vector field, which obeys
\begin{equation}
R\lrcorner \Omega =0 \hspace{1cm}\mathrm{and}\hspace{1cm} R\lrcorner \eta =1.
\end{equation}
In canonical coordinates we have $R=\frac{\partial}{\partial t}$.

For each $f\in C^{\infty}(M)$ is assigned a vector field $X_{f}$ on $M$, called the Hamiltonian vector field for $f$, according to
\begin{equation}
X_{f}\lrcorner \Omega =df-(Rf)\eta \hspace{1cm}\mathrm{and}\hspace{1cm} X_{f}\lrcorner \eta =0.
\end{equation}
In canonical coordinates we have
\begin{equation}
X_{f}=\frac{\partial f}{\partial p_{i}}\frac{\partial}{\partial q^{i}}-\frac{\partial f}{\partial q^{i}}\frac{\partial}{\partial p_{i}}.
\end{equation}
The assignment $f\longmapsto X_{f}$ is linear, that is
\begin{equation}
X_{f+\alpha g}=X_{f}+\alpha X_{g},
\end{equation}
$\forall f,g\in C^{\infty}(M)$ and $\forall \alpha \in\mathbb{R}$. Given $f,g \in C^{\infty}(M)$ the Poisson bracket of $f$ and $g$ is defined by
\begin{equation}
\lbrace f,g\rbrace=X_{g}f=\Omega(X_{f},X_{g}).
\end{equation}
In canonical coordinates, it reads
\begin{equation}
\lbrace f,g\rbrace=\frac{\partial f}{\partial q^{i}}\frac{\partial g}{\partial p_{i}}-\frac{\partial f}{\partial p_{i}}\frac{\partial g}{\partial q^{i}}.
\end{equation}

The theory of time-dependent Hamiltonian systems can be developed under the mathematical formalism of cosymplectic geometry (see \cite{LR89,CLL92,LS2017} Given $H\in C^{\infty}(M)$, the dynamics of the Hamiltonian system on $(M,\Omega,\eta)$ (the phase space) with Hamiltonian function $H$ is defined by the evolution vector field $E_{H}=X_{H}+R$ which is the solution to equations
\begin{equation}
E_{H}\lrcorner \Omega =dH-(RH)\eta \hspace{1cm}\mathrm{and}\hspace{1cm} E_{H}\lrcorner \eta =1.
\end{equation}
The trajectories $\psi(t)=(q^{1}(t),\cdots,q^{n}(t),p_{1}(t),\cdots,p_{n}(t),t)$ of the system are the integral curves of $E_{H}$, they satisfy the Hamilton equations of motion
\begin{equation}
\dot{q^{i}} =\frac{\partial H}{\partial p_{i}}, \hspace{1cm}
\dot{p_{i}} =-\frac{\partial H}{\partial q^{i}}.
\end{equation}

The evolution of a function $f\in C^{\infty}(M)$ (an observable) along the trajectories of the system is given by
\begin{equation}
\dot{f}= L_{E_{H}}f=E_{H}f=X_{H}f+Rf=\lbrace f,H\rbrace+\frac{\partial f}{\partial t}.
\end{equation}
We say that a function $f\in C^{\infty}(M)$ is a constant of motion of the system if it is constant along the trajectories of the system, that is, $f$ is a constant of motion if $L_{E_{H}}f=0$ ($\lbrace f,H\rbrace+\frac{\partial f}{\partial t}=0$).

Now we move to the formalism of contact Hamiltonian systems which naturally describe dissipative systems (for details see \cite{LS2017,LL2019,BCT2017}), and they also have applications in thermodynamics, statistical mechanics and others \cite{Bravetti2017}. 

\begin{de}
\label{deContact}
Let $M$ be a $2n+1$ dimensional smooth manifold. A contact structure on $M$ is a 1-form $\theta$ on $M$ such that $\theta\wedge d\theta^{n}\neq 0$. If $\theta$ is a contact structure on $M$ we say that $(M,\theta)$ is a contact manifold.
\end{de}

There is a wider notion of contact manifolds, some authors define a contact structure on a manifold as a one-codimensional maximally non-integrable distribution \cite{LL2020,BH2016,Geiges2009,GG2022}. Locally a contact structure can be expressed as the kernel of a one-form $\theta$ satisfying $\theta\wedge d\theta^{n}\neq 0$, $\theta$ is called a local contact form. However, not every contact structure admits a global contact form. When a contact structure admits a global contact form, the contact manifold is called co-oriented. Every contact manifold $(M,\theta)$ from definition \ref{deContact} is a co-oriented contact manifold in this wider sense by taking the distribution $Ker(\theta)$. In this paper we restrict ourselves to co-oriented contact manifold and we refer to them as contact manifolds. 

Let $(M,\theta)$ be a contact manifold of dimension $2n+1$. Around any point $p\in M$ there exist local coordinates $(q^{1},\cdots,q^{n},p_{1},\cdots,p_{n},z)$, i.e. Darboux coordinates, such that
\begin{equation}
\theta=dz-p_{i}dq^{i}.
\end{equation}
There exists a distinguished vector field $R$ on $M$, called the Reeb vector field, which obeys 
\begin{equation}
R\lrcorner \theta =1 \hspace{1cm}\mathrm{and}\hspace{1cm} R\lrcorner d\theta =0.
\end{equation}
In canonical coordinates we simply have $R=\frac{\partial}{\partial z}$.

For each $f\in C^{\infty}(M)$ is assigned a vector field $X_{f}$ on $M$, called the Hamiltonian vector field for $f$, according to
\begin{equation}
X_{f}\lrcorner \theta = -f \hspace{1cm}\mathrm{and}\hspace{1cm} X_{f}\lrcorner d\theta =df-(Rf)\theta.
\end{equation}
In canonical coordinates we have
\begin{equation}
X_{f}=\frac{\partial f}{\partial p_{i}}\frac{\partial}{\partial q^{i}}-\left( \frac{\partial f}{\partial q^{i}}+p_{i}\frac{\partial f}{\partial z}\right) \frac{\partial}{\partial p_{i}}+\left( p_{i}\frac{\partial f}{\partial p_{i}}-f\right)\frac{\partial}{\partial z}.
\end{equation}
It can be checked that the assignment $f\longmapsto X_{f}$ is linear, that is
\begin{equation}
X_{f+\alpha g}=X_{f}+\alpha X_{g},
\end{equation}
$\forall f,g\in C^{\infty}(M)$ and $\forall \alpha \in\mathbb{R}$.

Symplectic and cosymplectic manifolds are Poisson manifolds (symplectic and cosymplectic structures define Poisson brackets), but a contact manifold is strictly a Jacobi manifold, i.e., a contact structure on a manifold defines a Jacobi bracket. Given $f,g \in C^{\infty}(M)$ the Jacobi bracket of $f$ and $g$ is defined by
\begin{equation}
\lbrace f,g\rbrace=X_{g}f+fRg=\theta([X_{f},X_{g}]).
\end{equation}
In canonical coordinates we have
\begin{equation}
\lbrace f,g\rbrace=\frac{\partial f}{\partial q^{i}}\frac{\partial g}{\partial p_{i}}-\frac{\partial f}{\partial p_{i}}\frac{\partial g}{\partial q^{i}}+\frac{\partial f}{\partial z}\left( p_{i}\frac{\partial g}{\partial p_{i}}-g\right)-\frac{\partial g}{\partial z}\left( p_{i}\frac{\partial f}{\partial p_{i}}-f\right).
\end{equation}

Hamiltonian systems on contact manifolds are called contact Hamiltonian systems. Given $H\in C^{\infty}(M)$, the dynamics of the Hamiltonian on $(M,\theta)$ (the phase space) with Hamiltonian function $H$ is defined by the Hamiltonian vector field $X_{H}$.
In canonical coordinates we have
\begin{equation}
X_{H}=\frac{\partial H}{\partial p_{i}}\frac{\partial}{\partial q^{i}}-\left( \frac{\partial H}{\partial q^{i}}+p_{i}\frac{\partial H}{\partial z}\right) \frac{\partial}{\partial p_{i}}+\left( p_{i}\frac{\partial H}{\partial p_{i}}-H\right)\frac{\partial}{\partial z}.
\end{equation}
The trajectories $\psi(t)=(q^{1}(t),\cdots,q^{n}(t),p_{1}(t),\cdots,p_{n}(t),z(t))$ of the system are the integral curves of $X_{H}$, they satisfy the dissipative Hamilton equations of motion
\begin{equation}
\dot{q^{i}} =\frac{\partial H}{\partial p_{i}}, \hspace{1cm}
\dot{p_{i}} =-\left( \frac{\partial H}{\partial q^{i}}+p_{i}\frac{\partial H}{\partial z}\right) , \hspace{1cm}
\dot{z}=p_{i}\frac{\partial H}{\partial p_{i}}-H.
\end{equation}

The evolution of a function $f\in C^{\infty}(M)$ (an observable) along the trajectories of the system reads
\begin{equation}
\dot{f}= L_{X_{H}}f=X_{H}f=\lbrace f,H\rbrace-fRH.
\end{equation}
We say that a function $f\in C^{\infty}(M)$ is a constant of motion of the system if it is constant along the trajectories of the system, that is, $f$ is a constant of motion if $L_{X_{H}}f=0$ ($\lbrace f,H\rbrace-fRH=0$).

Finally we complete the notation and language employed in this paper. Here we briefly review the formalism of time-dependent contact Hamiltonian systems introduced in \cite{Letal2022}.

\begin{de}
Let $M$ be a $2n+2$ dimensional smooth manifold. A cocontact structure on $M$ is a couple $(\theta,\eta)$ of 1-forms on $M$ such that $\eta$ is closed and $\eta\wedge\theta\wedge(d\theta)^{n}\neq 0$. If $(\theta,\eta)$ is a cocontact structure on $M$ we say that $(M,\theta,\eta)$ is a cocontact manifold.
\end{de}

Let $(M,\theta,\eta)$ be a cocontact manifold of dimension $2n+2$. Around any point $p\in M$ there exist local coordinates $(t,q^{1},\cdots,q^{n},p_{1},\cdots,p_{n},z)$, called canonical coordinates or Darboux coordinates, such that
\begin{equation}
\theta=dz-p_{i}dq^{i}\hspace{1cm}\mathrm{and}\hspace{1cm}\eta=dt.
\end{equation}
There exist two distinguished vector fields $R_{z}$ and $R_{t}$ on $M$, called the contact Reeb vector field and the time Reeb vector field respectively, such that 
\begin{equation}
\left\lbrace \begin{array}{c}
R_{z}\lrcorner \eta=0, \\ 
R_{z}\lrcorner \theta=1, \\ 
R_{z}\lrcorner d\theta=0
\end{array} \right. 
\end{equation}
and
\begin{equation}
\left\lbrace \begin{array}{c}
R_{t}\lrcorner \eta=1, \\ 
R_{t}\lrcorner \theta=0, \\ 
R_{t}\lrcorner d\theta=0
\end{array} \right. .
\end{equation}
In canonical coordinates, we have $R_{z}=\frac{\partial}{\partial z}$ and $R_{t}=\frac{\partial}{\partial t}$.

For each $f\in C^{\infty}(M)$ is assigned a vector field $X_{f}$ on $M$, called the contact Hamiltonian vector field for $f$, according to
\begin{equation}
X_{f}\lrcorner \theta=-f,\hspace{1cm}X_{f}\lrcorner d\theta=df-(R_{z}f)\theta-(R_{t}f)\eta \hspace{1cm}\mathrm{and}\hspace{1cm} X_{f}\lrcorner \eta =0.
\end{equation}
In canonical coordinates, we have
\begin{equation}
X_{f}=\frac{\partial f}{\partial p_{i}}\frac{\partial}{\partial q^{i}}-\left( \frac{\partial f}{\partial q^{i}}+p_{i}\frac{\partial f}{\partial z}\right) \frac{\partial}{\partial p_{i}}+\left( p_{i}\frac{\partial f}{\partial p_{i}}-f\right)\frac{\partial}{\partial z}.
\end{equation}
We can observe that the assignment $f\longmapsto X_{f}$ is linear, that is
\begin{equation}
X_{f+\alpha g}=X_{f}+\alpha X_{g},
\end{equation}
$\forall f,g\in C^{\infty}(M)$ and $\forall \alpha \in\mathbb{R}$. Like contact manifolds, cocontact manifolds are Jacobi manifolds; given $f,g \in C^{\infty}(M)$ the Jacobi bracket of $f$ and $g$ is defined by
\begin{equation}
\lbrace f,g\rbrace=X_{g}f+fR_{z}g.
\end{equation}
In canonical coordinates, we have
\begin{equation}
\lbrace f,g\rbrace=\frac{\partial f}{\partial q^{i}}\frac{\partial g}{\partial p_{i}}-\frac{\partial f}{\partial p_{i}}\frac{\partial g}{\partial q^{i}}+\frac{\partial f}{\partial z}\left( p_{i}\frac{\partial g}{\partial p_{i}}-g\right)-\frac{\partial g}{\partial z}\left( p_{i}\frac{\partial f}{\partial p_{i}}-f\right).
\end{equation}

Time-dependent contact Hamiltonian systems are defined under the formalism of cocontact geometry. Given $H\in C^{\infty}(M)$, the dynamics of the Hamiltonian system on $(M,\theta,\eta)$ (the phase space) with Hamiltonian function $H$ is defined by the evolution vector field $E_{H}=X_{H}+R_{t}$ which is the solution to equations
\begin{equation}
E_{H}\lrcorner \theta=-H,\hspace{1cm}E_{H}\lrcorner d\theta=dH-(R_{z}H)\theta-(R_{t}H)\eta \hspace{1cm}\mathrm{and}\hspace{1cm} E_{H}\lrcorner \eta =1.
\end{equation}
In canonical coordinates
\begin{equation}
E_{H}=\frac{\partial H}{\partial p_{i}}\frac{\partial}{\partial q^{i}}-\left( \frac{\partial H}{\partial q^{i}}+p_{i}\frac{\partial H}{\partial z}\right) \frac{\partial}{\partial p_{i}}+\left( p_{i}\frac{\partial H}{\partial p_{i}}-H\right)\frac{\partial}{\partial z}+\frac{\partial}{\partial t}.
\end{equation}
The trajectories $\psi(s)=(t(s),q^{1}(s),\cdots,q^{n}(s),p_{1}(s),\cdots,p_{n}(s),z(s))$ of the system are the integral curves of $E_{H}$, they satisfy the equations
\begin{equation}
\dot{q^{i}} =\frac{\partial H}{\partial p_{i}}, \hspace{1cm}
\dot{p_{i}} =-\left( \frac{\partial H}{\partial q^{i}}+p_{i}\frac{\partial H}{\partial z}\right) , \hspace{1cm}
\dot{z}=p_{i}\frac{\partial H}{\partial p_{i}}-H, \hspace{1cm} \dot{t}=1.
\end{equation}
Since $\dot{t}=1$, we can take $t=s$, which implies that the temporal parameter for the system is $t$, that is, the trajectories of the system are parametrized by $t$
\begin{equation}
\psi(t)=(t,q^{1}(t),\cdots,q^{n}(t),p_{1}(t),\cdots,p_{n}(s),z(t))
\end{equation}
and they obey the dissipative Hamilton equations of motion 
\begin{equation}
\dot{q^{i}} =\frac{\partial H}{\partial p_{i}}, \hspace{1cm}
\dot{p_{i}} =-\left( \frac{\partial H}{\partial q^{i}}+p_{i}\frac{\partial H}{\partial z}\right) , \hspace{1cm}
\dot{z}=p_{i}\frac{\partial H}{\partial p_{i}}-H.
\end{equation}
The evolution of a function $f\in C^{\infty}(M)$ (an observable) along the trajectories of the system is given by
\begin{equation}
\dot{f}= L_{E_{H}}f=E_{H}f=X_{H}f+R_{t}f=\lbrace f,H\rbrace-fR_{z}H+R_{t}f.
\end{equation}
We say that a function $f\in C^{\infty}(M)$ is a constant of motion of the system if it is constant along the trajectories of the system, that is, $f$ is a constant of motion if $L_{E_{H}}f=0$ ($\lbrace f,H\rbrace-fR_{z}H+R_{t}f=0$).

\section{Lie algebras of symmetries for vector fields in $\mathbb{R}^{n}$}
\label{sec3}

Here we present a brief review of Lie's theorem for vector differential equations in $\mathbb{R}^{n}$. For this, let us consider a smooth vector field $v$ on $\mathbb{R}^{n}$.

\begin{de}
We say that a vector field $u$ on $\mathbb{R}^{n}$ is a symmetry of the vector field $v$ if $[u,v]=0$, where $[,]$ is the Lie bracket of vector fields.
\end{de}

Let us consider a dynamical system on $\mathbb{R}^{n}$ defined by the vector field $v$. If we consider local coordinates $x=(x^{1},\ldots,x^{n})$ on $\mathbb{R}^{n}$ then the equation of motion of the system is given by the vector differential equation $\dot{x}=v(x)$. If we write $v$ in terms of the coordinate basis vectors $\frac{\partial}{\partial x^{i}}$ as
\begin{equation}
v(x)=v^{1}(x)\frac{\partial}{\partial x^{1}}+\cdots+v^{n}(x)\frac{\partial}{\partial x^{n}}, 
\end{equation}
then the equation $\dot{x}=v(x)$ is equivalent to the following system of differential equations
\begin{equation}\label{eq1}
\left\lbrace \begin{array}{c}
\dot{x}^{1}=v^{1}(x^{1},\ldots,x^{n}), \\
\vdots \\
\dot{x}^{n}=v^{n}(x^{1},\ldots,x^{n}).
\end{array} \right.
\end{equation}
Let $u$ be a smooth vector field on $\mathbb{R}^{n}$. It follows from the Straightening Out Theorem \cite{AMR88,Arnold92} that there are coordinates $(y^{1},\ldots,y^{n})$ on $\mathbb{R}^{n}$ such that $u(y)=\frac{\partial}{\partial y^{n}}$. Observe that if $u$ is a symmetry of $v$, that is $[u,v]=0$, then the component functions $\overline{v}^{1},\ldots,\overline{v}^{n}$ of $v$ in terms of the coordinate basis vectors $\frac{\partial}{\partial y^{i}}$ do not depend on the coordinate $y^{n}$, indeed,
\begin{equation}
0=[u,v]=\left[ \frac{\partial}{\partial y^{n}},\overline{v}^{i}\frac{\partial}{\partial y^{i}}\right] =\frac{\partial \overline{v}^{i}}{\partial y^{n}}\frac{\partial}{\partial y^{i}},
\end{equation}
therefore $\frac{\partial \overline{v}^{i}}{\partial y^{n}}=0$, i.e. $\overline{v}^{i}=\overline{v}^{i}(y^{1},\ldots,y^{n-1})$ for $i=1,\ldots,n$. So we have that the equation of motion of the system is equivalent to the system of differential equations 
\begin{equation}\label{eq2}
\left\lbrace \begin{array}{c}
\dot{y}^{1}=\overline{v}^{1}(y^{1},\ldots,y^{n-1}), \\ 
\vdots\\
\dot{y}^{n-1}=\overline{v}^{n-1}(y^{1},\ldots,y^{n-1}),\\ 
\dot{y}^{n}=\overline{v}^{n}(y^{1},\ldots,y^{n-1}).
\end{array} \right.
\end{equation}
Now the problem of solving the equation of motion of the system reduces to solving the following system of $n-1$ differential equations
\begin{equation}\label{eq3}
\left\lbrace \begin{array}{c}
\dot{y}^{1}=\overline{v}^{1}(y^{1},\ldots,y^{n-1}), \\ 
\vdots\\ 
\dot{y}^{n-1}=\overline{v}^{n-1}(y^{1},\ldots,y^{n-1}),
\end{array} \right.
\end{equation}
and integrating the differential equation $\dot{y}^{n}=\overline{v}^{n}(y^{1},\ldots,y^{n-1})$. Note that if we could repeat this idea on system (\ref{eq3}), we would obtain a system of $n-2$ differential equations and two integrable differential equations. That is possible if we have another symmetry $w$ of the vector field $v$ such that $[w,u]=\lambda u$, where $\lambda$ is a non zero scalar. Indeed, first let us see that the component functions of the smooth vector field $w$ do not depend on the coordinate $y^{n}$, 
\begin{equation}
\begin{split}
\lambda u=[w,u] & \Longrightarrow  \lambda \frac{\partial}{\partial y^{n}}=\left[ w^{i}\frac{\partial}{\partial y^{i}},\frac{\partial}{\partial y^{n}}\right] \\
& \Longrightarrow  \lambda \frac{\partial}{\partial y^{n}}= - \frac{\partial w^{i}}{\partial y^{n}}\frac{\partial}{\partial y^{i}}\\
& \Longrightarrow  \frac{\partial w^{i}}{\partial y^{n}}=0, \hspace{.3cm} \forall i\in \lbrace 1,\ldots,n-1\rbrace,
\end{split}
\end{equation}
so that $w^{i}=w^{i}(y^{1},\ldots,y^{n-1})$ for $i=1,\ldots, n-1$; thus the smooth vector field $w^{1}\frac{\partial}{\partial y^{1}}+\cdots+w^{n-1}\frac{\partial}{\partial y^{n-1}}$ is a symmetry of the smooth vector field $\overline{v}^{1}\frac{\partial}{\partial y^{1}}+\cdots+\overline{v}^{n-1}\frac{\partial}{\partial y^{n-1}}$ and we can apply the previous process to reduce system (\ref{eq3}) to a system of $n-2$ differential equations and one integrable differential equation, which leads us to a system of $n-2$ differential equations and two integrable differential equations. By repeating this process a sufficient number of times, we obtain the solution of the original system of differential equations just by integrating some differential equations and doing some algebraic operations (because of the change of coordinates). 

On the other hand, we say that a system of differential equations is integrable by quadratures if we can solve it by integrating some known functions and doing some algebraic operations \cite{AKN2006}. The Lie's theorem \cite{AKN2006,CFGR2015} establishes sufficient conditions for a system defined by a vector differential equation of the form $\dot{x}=v(x)$ to be integrable by quadratures, it reads as follows.

\begin{te}\label{theorem1}
Let $u_{1},\ldots u_{n}$ be linearly independent smooth vector fields on $\mathbb{R}^{n}$. If $u_{1},\ldots,u_{n}$ are symmetries of the vector field $v$ and they generate a solvable Lie algebra with the Lie bracket $[,]$ of vector fields, then the system defined by $\dot{x}=v(x)$ is integrable by quadratures.
\end{te}

For our purposes we need to consider smooth vector fields $v$ on $\mathbb{R}^{n}\times\mathbb{R}$ of the form $v(x,t)=v^{1}(x,t)\frac{\partial}{\partial x^{1}}+\cdots+v^{n}(x,t)\frac{\partial}{\partial x^{n}}+\frac{\partial}{\partial t}$, with $t$ the coordinate in $\mathbb{R}$. If we consider the dynamical system on $\mathbb{R}^{n}\times\mathbb{R}$ defined by the vector field $v$, then the equations of motion are 
\begin{equation}\label{eqt}
\left\lbrace \begin{array}{c}
\dot{x}^{1}=v^{1}(x^{1},\ldots,x^{n},t), \\
\vdots \\
\dot{x}^{n}=v^{n}(x^{1},\ldots,x^{n},t),\\
\dot{t}=1
\end{array} \right.
\end{equation}
By integrating the last equation we have the solution for the variable $t$, then we only have to solve a system of $n$ differential equations, so to apply the reduction procedure described previously we only need $n$ independent symmetries of $v$ that generate a solvable Lie algebra with the Lie bracket of vector fields.

\section{Lie algebras of symmetries for symplectic and cosymplectic Hamiltonian systems}
\label{sec4}

In \cite{AKN2006,Kozlov96} it is showed for autonomous and nonautonomous Hamiltonian systems under the phase space $\mathbb{R}^{2n}$ and $\mathbb{R}^{2n}\times\mathbb{R}$ respectively, by using the results of Lie, that having a solvable Lie algebra of constants of motion of a Hamiltonian system allows us to find solutions of the Hamilton equations of motion by quadratures; in \cite{Azuaje2022} this result is extended to time-dependent Hamiltonian systems under the phase space $M\times\mathbb{R}$ with $M$ a symplectic manifold. We show this result for time-independent Hamiltonian systems under symplectic manifolds and for time-dependent Hamiltonian systems under cosymplectic manifolds. For the first case we have the following theorem.

\begin{te}\label{theorem2}
Let $(M,\omega,H)$ be a time-independent Hamiltonian system with $(M,\omega)$ a symplectic manifold of dimension $2n$. If we have $n$ constants of motion $f_{1},\ldots,f_{n}$ such that
\begin{enumerate}
\item $\lbrace f_{i},f_{j}\rbrace =c_{ij}^{k}f_{k}$, with $c_{ij}^{k}\in\mathbb{R}$,
\item the Lie algebra generated by $f_{1},\ldots,f_{n}$ is a solvable Lie algebra with the Lie bracket defined by the Poisson bracket,
\item on the set $M_{f}=\lbrace x\in M : f_{i}(x)=\alpha_{i}, \alpha_{i}\in \mathbb{R}\rbrace$ the functions $f_{1},\ldots,f_{n}$ are functionally independent and, 
\item $c_{ij}^{k}\alpha_{k}=0$ for $i,j=1,\ldots,n$,
\end{enumerate}
then $M_{f}$ is a smooth submanifold of $M$ of dimension $n$, and the solutions of the Hamilton equations that live in $M_{f}$ can be found by quadratures.
\end{te}

Before showing the proof of this theorem, we present a brief review of the concept of functional independence and its consequences. Consider a smooth manifold $N$, we say that the functions $f_{1},f_{2},\cdots,f_{k}$ with $k\leq dim(N)$ are functionally independent at the point $p\in N$ if the differential maps $df_{1}|_{p},df_{2}|_{p},\ldots,df_{k}|_{p}$ are linearly independent; this means that $p$ is a regular point of the differentiable function $F:N\longrightarrow \mathbb{R}^{k}$ defined by $F=(f_{1},f_{2},\ldots,f_{k})$. Consider the set $N_{f}=\lbrace x\in N: f_{i}(x)=\alpha_{i}, \alpha_{i}\in \mathbb{R}\rbrace$, and suppose that $f_{1},f_{2},\ldots,f_{k}$ are functionally independent, from the Regular Level Set Theorem \cite{Lee2012} we have that $N_{f}$ is a smooth submanifold of $N$ of dimension $dim(N)-k$.

Proof: From the functional independence of the functions $f_{1},\ldots,f_{n}$ on $M_{f}$, we have that $M_{f}$ is a smooth submanifold of $M$ of dimension $2n-n=n$. We know that the assignment $f\mapsto X_{f}$ is a Lie algebra antihomomorphism between the Lie algebras $(C^{\infty}(M),\lbrace,\rbrace)$ and $(\mathfrak{X}(M),[,])$ \cite{Lee2012,BG2021}, i.e. $X_{\lbrace f,g\rbrace}=-[X_{f},X_{g}]$, so since the functions $f_{1},\ldots,f_{n}$ are constants of motion, the vector fields $X_{f_{1}},\cdots,X_{f_{n}}$ are symmetries of the Hamiltonian vector field $X_{H}$ and generate a solvable Lie algebra with the Lie bracket of vector fields. In addition we have that the vector fields $X_{f_{1}},\ldots,X_{f_{n}}$ are tangent to the submanifold $M_{f}$ because $X_{f_{i}}f_{j}=\lbrace f_{j},f_{i}\rbrace$ which is zero on $M_{f}$, also the Hamiltonian vector field $X_{H}$ is tangent to the submanifold $M_{f}$ since $X_{H}f_{i}=\lbrace f_{i},H\rbrace=0$ for $i=1,\ldots,n$. Then the solutions of the Hamilton equations that live on $M_{f}$ are solutions of the vector differential equation $\dot{x}=X_{H}(x)$ with $(x^{1},\ldots,x^{n})$ local coordinates on $M_{f}$.  Since the vector fields $X_{f_{1}},\ldots,X_{f_{n}}$ generate a solvable Lie algebra of symmetries of the vector field $X_{H}$, then the system defined by the equation $\dot{x}=X_{H}(x)$ is solvable by quadratures (Theorem \ref{theorem1}), that is, the solutions of the Hamilton equations that live in $M_{f}$ can be found by quadratures.
\rule{5pt}{5pt}

In the cosymplectic framework we have an analogous result.

\begin{te}\label{theorem3}
Let $(M,\Omega,\eta,H)$ be a time-dependent Hamiltonian system with $(M,\Omega,\eta)$ a cosymplectic manifold of dimension $2n+1$. If we have $n$ constants of motion $f_{1},\ldots,f_{n}$ such that
\begin{enumerate}
\item $\lbrace f_{i},f_{j}\rbrace =c_{ij}^{k}f_{k}$, with $c_{ij}^{k}\in\mathbb{R}$,
\item the Lie algebra generated by $f_{1},\ldots,f_{n}$ is a solvable Lie algebra with the Lie bracket defined by the Poisson bracket,
\item on the set $M_{f}=\lbrace x\in M : f_{i}(x)=\alpha_{i}, \alpha_{i}\in \mathbb{R}\rbrace$ the functions $f_{1},\ldots,f_{n}$ are functionally independent and, 
\item $c_{ij}^{k}\alpha_{k}=0$ for $i,j=1,\ldots,n$,
\end{enumerate}
then $M_{f}$ is a smooth submanifold of $M$ of dimension $n+1$, and the solutions of the Hamilton equations that live in $M_{f}$ can be found by quadratures.
\end{te}

Our procedure is analogous to the one presented in \cite{Azuaje2022}. First, we have an analogous to Lemma 1 in \cite{Azuaje2022} that is fundamental for the proof of theorem \ref{theorem3}.

\begin{lem}\label{lem1}
If $f_{1},f_{2},\ldots,f_{k}:M \longrightarrow \mathbb{R}$ are functionally independent constants of motion, then around any point in $M_{f}$ defined as in Theorem \ref{theorem3}, we can find local coordinates $(y^{1},\cdots,y^{2n-k},t)$.
\end{lem}
The proof is analogous to the presented in \cite{Azuaje2022}, there are only some subtle but essential differences.

Proof: Since the functions $f_{1},f_{2},\ldots,f_{k}$ are functionally independent, from the Regular Level set Theorem \cite{Lee2012} we have that $M_{f}$ is a smooth submanifold of $M$ of dimension $2n+1-k$.

Consider $p\in M_{f}$. Let $(U,\phi)$ be a smooth chart on $M$ around $p$, with $U$ an open subset of $\mathbb{R}^{2n+1}$. We have that the the total derivative of $f$ at $p$ is a matrix with $k$ linearly independent columns, those columns are some of the first $2n$ ones since the last column of the matrix is a linear combination of the others. Let $x^{1},\ldots,x^{k}\in \lbrace q^{1},\ldots,q^{n},p_{1},\ldots,p_{n}\rbrace$ be such that the columns
\begin{equation}
\left( \begin{array}{c}
\frac{\partial f_{1}}{\partial x^{1}}(p) \\ 
\vdots \\ 
\frac{\partial f_{k}}{\partial x^{1}}(p)
\end{array} \right), \cdots ,\left( \begin{array}{c}
\frac{\partial f_{1}}{\partial x^{k}}(p) \\ 
\vdots \\ 
\frac{\partial f_{k}}{\partial x^{k}}(p)
\end{array} \right),
\end{equation}
 are linearly independent, so the matrix
\begin{equation}
\left( \begin{array}{ccc}
\frac{\partial f_{1}}{\partial x^{1}}(p) & \cdots & \frac{\partial f_{1}}{\partial x^{k}}(p) \\ 
\vdots &  & \vdots \\ 
\frac{\partial f_{k}}{\partial x^{1}}(p) & \cdots & \frac{\partial f_{k}}{\partial x^{k}}(p)
\end{array} \right)
\end{equation}
is no singular, then there exist an open subset $U_{0}\subseteq \mathbb{R}^{2n+1-k}$, an open subset $W_{0}\subseteq \mathbb{R}^{k}$ and a differentiable function $g:U_{0}\longrightarrow W_{0}$ such that
\begin{equation}
\phi^{-1}(p)\in U_{0}\times W_{0}\subseteq U,
\end{equation}
and on $M_{f}\cap \phi (U_{0}\times W_{0})$ we have that $x=g(y,t)$, where $x=(x^{1},\ldots,x^{k})$ and $y=(y^{1},\ldots,y^{2n-k})$ are such that $(y^{1},\ldots,y^{2n-k},x^{1},\ldots,x^{k})=(q^{1},\ldots,q^{n},p_{1},\ldots,p_{n})$. So we obtain a smooth coordinate chart $(U_{0},\varphi)$ on $M_{f}$ with $p\in \varphi (U_{0})$ and $\varphi$ defined by
\begin{equation}
\varphi(y^{1},\ldots,y^{2n-k},t)=\phi(y^{1},\ldots,y^{2n-k},g(y^{1},\ldots,y^{2n-k},t),t).
\end{equation}
\rule{5pt}{5pt}

As in the proof of theorem \ref{theorem2}, in the cosymplectic framework we have that constants of motion of the system are related to symmetries of  $E_{H}$ in the following sense, if $f$ is a constant of motion then the Hamiltonian vector field $X_{f}$ for $f$ is a symmetry of $E_{H}$, indeed, let us see that the assignment $f\mapsto X_{f}$ is a Lie algebra antihomomorphism between the Lie algebras $(C^{\infty}(M),\lbrace,\rbrace)$ and $(\mathfrak{X}(M),[,])$, i.e. $X_{\lbrace f,g\rbrace}=-[X_{f},X_{g}]$; in addition we have $X_{Rf}=-[X_{f},R]$. Indeed, given $f,g,h\in C^{\infty}(M)$ we have
\begin{equation}
\begin{split}
-[X_{f},X_{g}]h&=-X_{f}X_{g}h+X_{g}X_{f}h\\
&=-X_{f}\lbrace h,g\rbrace+X_{g}\lbrace h,f\rbrace\\
&=-\lbrace \lbrace h,g\rbrace,f\rbrace+\lbrace \lbrace h,f\rbrace,g\rbrace\\
&=\lbrace \lbrace g,h\rbrace,f\rbrace+\lbrace \lbrace h,f\rbrace,g\rbrace\\
&=-\lbrace \lbrace f,g\rbrace, h\rbrace \hspace{1cm}(\textit{Jacoby identity})\\
&=X_{\lbrace f,g\rbrace}h.
\end{split}
\end{equation}
On the other hand
\begin{equation}
-[X_{f},R]\lrcorner\eta=\eta([R,X_{f}])=R(\eta(X_{f}))-X_{f}(\eta(R))=-X_{f}(1)=0
\end{equation}
and for all $X\in\mathfrak{X}(M)$ we have
\begin{equation}
\begin{split}
(-[X_{f},R]\lrcorner\Omega)(X)&=\Omega([R,X_{f}],X)\\
&=X\Omega(R,X_{f})-R\Omega(X,X_{f})+X_{f}\Omega(X,R)-\Omega([X,R],X_{f})+\Omega([X,X_{f}],R)\\
&=R\Omega(X_{f},X)+\Omega(X_{f},[X,R])\\
&=R(df(X)-(Rf)\eta(X))+df([X,R])-(Rf)\eta([X,R])\\
&=R(df(X))-(RRf)\eta(X)-(Rf)R(\eta(X))+X(df(R))-R(df(X))\\
&-(Rf)(X(\eta(R))-R(\eta(X)))\\
&=-RRf\eta(X)+X(R(f))\\
&=d(Rf)(X)-(RRf)\eta(X),
\end{split}
\end{equation}
therefore $-[X_{f},R]\lrcorner\Omega=d(Rf)-(RRf)\eta$, so that $X_{Rf}=-[X_{f},R]$.

Now we know that
\begin{equation}
[E_{H},X_{f}]=[X_{H}+R,X_{f}]=[X_{H},X_{f}]+[R,X_{f}],
\end{equation}
so that
\begin{equation}
\begin{split}
[E_{H},X_{f}]&=X_{\lbrace f,H\rbrace}+X_{\frac{\partial f}{\partial t}}\\
&= X_{\lbrace f,H\rbrace+\frac{\partial f}{\partial t}}\\
&=0.
\end{split}
\end{equation}

Now the proof of theorem \ref{theorem3} follows from theorem \ref{theorem1} and lemma \ref{lem1}.

Proof: From Lemma \ref{lem1} we have that the set $M_{f}$ is a smooth submanifold of $M$ of dimension $n+1$ and  around any point in $M_{f}$ we can find local coordinates $(y^{1},\ldots,y^{n},t)$. An essential observation is that the vector fields $X_{f_{1}},\ldots,X_{f_{n}}$ are tangent to the submanifold $M_{f}$ because $X_{f_{i}}(f_{j})=\lbrace f_{j},f_{i}\rbrace$ which is zero on $M_{f}$, also the evolution vector field $E_{H}$ is tangent to the submanifold $M_{f}$ since $E_{H}f_{i}=\lbrace f_{i},H\rbrace+\frac{\partial f_{i}}{\partial t}=0$ for $i=1,\ldots,n$. Then the solutions of the Hamilton equations that live in $M_{f}$ are solutions of the vector differential equation $\dot{y}=E_{H}(y,t)$, so that the solutions of the Hamilton equations that live in $M_{f}$ can be found by quadratures if the vector differential equation $\dot{y}=E_{H}(y,t)$ can be solved by quadratures; in fact we have that the equation $\dot{y}=E_{H}(y,t)$  satisfies the hypothesis of theorem \ref{theorem1}, indeed, the vector fields $X_{f_{1}},\ldots,X_{f_{n}}$ generate a solvable Lie algebra of symmetries of the vector field $E_{H}$.
\rule{5pt}{5pt}

It has been already mentioned that every Abelian Lie algebra is trivially a solvable Lie algebra, so every integrable system satisfies the hypothesis of the previous theorems presented in this section; for example for the Harmonic oscillator it is well known the existence of independent constants of motion that are pairwise in involution, so they form an Abelian Lie algebra over the Poisson bracket which is trivially a solvable Lie algebra, i.e., for the Harmonic oscillator we can apply theorem \ref{theorem2} or theorem \ref{theorem3} depending on whether the Hamiltonian function is time-independent or time-dependent. In \cite{Prykarpatsky99} we can find an example of a nonabelian Lie algebra of constants of motion for the motion of three particles on line $\mathbb{R}$ under a uniform potential field.

\section{Lie algebras of symmetries for contact and cocontact Hamiltonian systems}
\label{sec5}

The notion of integrability for contact Hamiltonian systems is a bit different from the one for symplectic or cosymplectic Hamiltonian systems; we know that there are important geometric differences between the phase spaces of contact Hamiltonian systems and symplectic or cosymplectic Hamiltonian systems, namely, symplectic and cosymplectic manifolds are Poisson manifolds, but a contact manifold is strictly a Jacobi manifold. The following definition is presented in \cite{Boyer2011,Visinescu2017}.
\begin{de}
We say that a contact Hamiltonian system $(M,\theta,H)$ is good if $H$ is constant along the flow of the Reeb vector field $R$ ($RH=0$), or equivalently, $H$ is a constant of motion. 
\end{de}

\begin{te}\label{theorem4}
Let $(M,\theta,H)$ be a good contact Hamiltonian system with $(M,\theta)$ a contact manifold of dimension $2n+1$. If we have $n$ constants of motion $f_{1},\ldots,f_{n}$ with $Rf=0$ (the Hamiltonian $H$ might be one of them) such that
\begin{enumerate}
\item $\lbrace f_{i},f_{j}\rbrace =c_{ij}^{k}f_{k}$, with $c_{ij}^{k}\in\mathbb{R}$,
\item the Lie algebra generated by $f_{1},\ldots,f_{n}$ is a solvable Lie algebra with the Lie bracket defined by the Jacobi bracket,
\item on the set $M_{f}=\lbrace x\in M : f_{i}(x)=\alpha_{i}, \alpha_{i}\in \mathbb{R}\rbrace$ the functions $f_{1},\ldots,f_{n}$ are functionally independent and, 
\item $c_{ij}^{k}\alpha_{k}=0$ for $i,j=1,\ldots,n$,
\end{enumerate}
then $M_{f}$ is a smooth submanifold of $M$ of dimension $n+1$, and the solutions of the Hamilton equations that live in $M_{f}$ can be found by quadratures.
\end{te}

Proof: Let $(q^{1},\cdots,q^{n},p_{1},\cdots,p_{n},z)$ be canonical coordinates for $(M,\theta,H)$. Since $RH=0$, i.e. $\frac{\partial H}{\partial z}=0$, then the Hamilton equations of motion are
\begin{equation}\label{eq4}
\left\lbrace \begin{array}{c}
\dot{q}^{1}=\frac{\partial H}{\partial p_{1}}(q^{1},\cdots,q^{n},p_{1},\cdots,p_{n}), \\
\vdots \\
\dot{q}^{n}=\frac{\partial H}{\partial p_{n}}(q^{1},\cdots,q^{n},p_{1},\cdots,p_{n}), \\
\dot{p}_{1}=-\frac{\partial H}{\partial q^{1}}(q^{1},\cdots,q^{n},p_{1},\cdots,p_{n}), \\
\vdots \\
\dot{p}_{n}=-\frac{\partial H}{\partial q^{n}}(q^{1},\cdots,q^{n},p_{1},\cdots,p_{n}),\\
\dot{z}=p_{i}\frac{\partial H}{\partial p_{i}}-H.
\end{array} \right.
\end{equation}

From the functional independence of the functions $f_{1},\ldots,f_{n}$ on $M_{f}$, we have that $M_{f}$ is a smooth submanifold of $M$ of dimension $2n+1-n=n+1$ and since $Rf_{i}=0$ we can find local coordinates on $M_{f}$ of the form $(x^{1},\ldots,x^{n},z)$. Again as in the symplectic and cosymplectic frameworks, we have that if $f$ is a constant of motion of the good contact Hamiltonian $(M,\theta,H)$ with $Rf=0$ then the Hamiltonian vector field $X_{f}$ for $f$ is a symmetry of $X_{H}$, indeed, for all $f,g\in C^{\infty}(M)$ with $Rf=Rg=0$ we have
\begin{equation}
\begin{split}
R(\lbrace f,g\rbrace) &=R(X_{g}f)\\
&=R(df(X_{g}))\\
&=R(d\theta(X_{f},X_{g}))\\
&=(L_{R}d\theta)(X_{f},X_{g})+d\theta(L_{R}X_{f},X_{g})+d\theta(X_{f},L_{R}X_{g})\\
&=-dg(L_{R}X_{f})+df(L_{R}X_{g})\\
&=dg(L_{X_{f}}R)-df(L_{X_{g}}R)\\
&=L_{X_{f}}(dg(R))-(L_{X_{f}}dg)(R)-L_{X_{g}}(df(R))+(L_{X_{g}}df)(R)\\
&=-d(X_{f}g)(R)+d(X_{g}f)(R)\\
&=-R(\lbrace g,f\rbrace)+R(\lbrace f,g\rbrace)\\
&=2R(\lbrace f,g\rbrace),
\end{split}
\end{equation}
therefore $R(\lbrace f,g\rbrace)=2R(\lbrace f,g\rbrace)$, so $R(\lbrace f,g\rbrace)=0$. And
\begin{equation}
\begin{split}
-[X_{f},X_{g}]h&=-X_{f}X_{g}h+X_{g}X_{f}h\\
&=-X_{f}\lbrace h,g\rbrace+X_{g}\lbrace h,f\rbrace\\
&=-\lbrace \lbrace h,g\rbrace,f\rbrace+\lbrace \lbrace h,f\rbrace,g\rbrace\\
&=\lbrace \lbrace g,h\rbrace,f\rbrace+\lbrace \lbrace h,f\rbrace,g\rbrace\\
&=-\lbrace \lbrace f,g\rbrace, h\rbrace \hspace{1cm}(\textit{Jacoby identity})\\
&=X_{\lbrace f,g\rbrace}h.
\end{split}
\end{equation}
So we conclude that $X_{\lbrace f,g\rbrace}=-[X_{f},X_{g}]$. So since the functions $f_{1},\ldots,f_{n}$ are constants of motion with $Rf_{i}=0$ then the vector fields $X_{f_{1}},\cdots,X_{f_{n}}$ are symmetries of the Hamiltonian vector field $X_{H}$; we also have that $R$ is a symmetry of $X_{H}$ since $[R,X_{H}]=X_{\lbrace H,1\rbrace}=0$ ($R$ is the Hamiltonian vector field for the function $1$) and $R,X_{f_{1}},\cdots,X_{f_{n}}$ generate a solvable Lie algebra with the Lie bracket of vector fields. In addition we have that the vector fields $R,X_{f_{1}},\ldots,X_{f_{n}}$ are tangent to the submanifold $M_{f}$ because $X_{f_{i}}f_{j}=\lbrace f_{j},f_{i}\rbrace$ which is zero on $M_{f}$ and $Rf_{i}=0$. Then the solutions of the Hamilton equations that live on $M_{f}$ are solutions of the vector differential equation $\dot{x}=X_{H}(x)$ with $(x^{1},\ldots,x^{n+1})$ local coordinates on $M_{f}$ with $x^{n+1}=z$.  Since the vector fields $R,X_{f_{1}},\ldots,X_{f_{n}}$ generate a solvable Lie algebra of symmetries of the vector field $X_{H}$, then the system defined by the equation $\dot{x}=X_{H}(x)$ is solvable by quadratures (theorem \ref{theorem1}), that is, the solutions of the Hamilton equations that live in $M_{f}$ can be found by quadratures.
\rule{5pt}{5pt}

Now let us see the time-dependent contact case. Again we restrict ourselves to good Hamiltonian systems.
\begin{de}
We say that a cocontact Hamiltonian system $(M,\theta,\eta,H)$ is good if $H$ is constant along the flow of the contact Reeb vector field $R_{z}$ ($R_{z}H=0$).
\end{de}

Observe that as in the cosymplectic case, the Hamiltonian function $H$ is not necessarily a constant of motion (the condition $R_{t}H=0$ is not required). 

\begin{te}\label{theorem5}
Let $(M,\theta,\eta,H)$ be a good cocontact Hamiltonian system with $(M,\theta,\eta)$ a cocontact manifold of dimension $2n+2$. If we have $n$ constants of motion $f_{1},\ldots,f_{n}$ with $R_{z}f_{i}=0$ such that
\begin{enumerate}
\item $\lbrace f_{i},f_{j}\rbrace =c_{ij}^{k}f_{k}$, with $c_{ij}^{k}\in\mathbb{R}$,
\item the Lie algebra generated by $f_{1},\ldots,f_{n}$ is a solvable Lie algebra with the Lie bracket defined by the Poisson bracket,
\item on the set $M_{f}=\lbrace x\in M : f_{i}(x)=\alpha_{i}, \alpha_{i}\in \mathbb{R}\rbrace$ the functions $f_{1},\ldots,f_{n}$ are functionally independent and, 
\item $c_{ij}^{k}\alpha_{k}=0$ for $i,j=1,\ldots,n$,
\end{enumerate}
then $M_{f}$ is a smooth submanifold of $M$ of dimension $n+2$, and the solutions of the Hamilton equations that live in $M_{f}$ can be found by quadratures.
\end{te}

The following lemma, analogous to lemma \ref{lem1}, is essential for the proof of theorem \ref{theorem5}.
\begin{lem}\label{lem2}
If $f_{1},f_{2},\ldots,f_{k}:M \longrightarrow \mathbb{R}$ are functionally independent constants of motion with $R_{z}f_{i}=0$, then around any point in $M_{f}$ defined as in Theorem \ref{theorem5}, we can find local coordinates $(t,y^{1},\cdots,y^{2n-k},z)$.
\end{lem}

Proof: From the functional independence of the functions $f_{1},\ldots,f_{k}$ on $M_{f}$, we have that $M_{f}$ is a smooth submanifold of $M$ of dimension $2n+2-k$. Let $(t,q^{1},\cdots,q^{n},p_{1},\cdots,p_{n},z)$ be canonical coordinates for $(M,\theta,\eta,H)$, then since $\frac{\partial f_{i}}{\partial z}=0$ ($R_{z}f_{i}=0$) then $\frac{\partial f_{i}}{\partial q^{j}}\frac{\partial H}{\partial p_{j}}-\frac{\partial f_{i}}{\partial p_{j}}\frac{\partial H}{\partial q^{j}}+\frac{\partial f_{i}}{\partial t}=0$ ($\lbrace f_{i},H\rbrace+R_{t}f=0$), so we have an analogous development to the proof of lemma \ref{lem1} and we conclude that on $M_{f}$ we can find local coordinates $(t,y^{1},\cdots,y^{2n+1-k})$. In addition as in the contact case since $\frac{\partial f_{i}}{\partial z}=0$ then on $M_{f}$ we have local coordinates $(t,y^{1},\ldots,y^{2n-k},z)$.
\rule{5pt}{5pt}

Now the proof of theorem \ref{theorem5} follows: 

Proof: From the functional independence of the functions $f_{1},\ldots,f_{n}$ on $M_{f}$, we have that $M_{f}$ is a smooth submanifold of $M$ of dimension $2n+2-n=n+2$ and we have local coordinates $(t,x^{1},\ldots,x^{n},z)$. We have that the Hamiltonian vector fields $X_{f_{i}}$ for $f_{i}$ is a symmetry of $E_{H}$, indeed, let us see that  $X_{\lbrace f,g\rbrace}=-[X_{f},X_{g}]$ and $X_{R_{t}f}=-[X_{f},R_{t}]$ for all $f,g\in C^{\infty}(M)$ with $R_{z}f=R_{z}g=0$, indeed,
\begin{equation}
\begin{split}
R_{z}(\lbrace f,g\rbrace) &=R_{z}(X_{g}f)\\
&=R_{z}(df(X_{g}))\\
&=R_{z}(d\theta(X_{f},X_{g})+(R_{t}f)\eta(X_{g}))\\
&=(L_{R_{z}}d\theta)(X_{f},X_{g})+d\theta(L_{R_{z}}X_{f},X_{g})+d\theta(X_{f},L_{R_{z}}X_{g})\\
&=-dg(L_{R_{z}}X_{f})+(R_{t}g)\eta(L_{R_{z}}X_{f})+df(L_{R_{z}}X_{g})-(R_{t}f)\eta(L_{R_{z}}X_{g})\\
&=dg(L_{X_{f}}R_{z})-df(L_{X_{g}}R_{z})\\
&=L_{X_{f}}(dg(R_{z}))-(L_{X_{f}}dg)(R_{z})-L_{X_{g}}(df(R_{z}))+(L_{X_{g}}df)(R_{z})\\
&=-d(X_{f}g)(R_{z})+d(X_{g}f)(R_{z})\\
&=-R_{z}(\lbrace g,f\rbrace)+R_{z}(\lbrace f,g\rbrace)\\
&=2R_{z}(\lbrace f,g\rbrace),
\end{split}
\end{equation}
therefore $R_{z}(\lbrace f,g\rbrace)=2R_{z}(\lbrace f,g\rbrace)$, so $R_{z}(\lbrace f,g\rbrace)=0$. And
\begin{equation}
\begin{split}
-[X_{f},X_{g}]h&=-X_{f}X_{g}h+X_{g}X_{f}h\\
&=-X_{f}\lbrace h,g\rbrace+X_{g}\lbrace h,f\rbrace\\
&=-\lbrace \lbrace h,g\rbrace,f\rbrace+\lbrace \lbrace h,f\rbrace,g\rbrace\\
&=\lbrace \lbrace g,h\rbrace,f\rbrace+\lbrace \lbrace h,f\rbrace,g\rbrace\\
&=-\lbrace \lbrace f,g\rbrace, h\rbrace \hspace{1cm}(\textit{Jacoby identity})\\
&=X_{\lbrace f,g\rbrace}h.
\end{split}
\end{equation}
So we conclude that $X_{\lbrace f,g\rbrace}=-[X_{f},X_{g}]$. On the other hand 
\begin{equation}
-[X_{f},R_{t}]\lrcorner\eta=\eta([R_{t},X_{f}])=R_{t}(\eta(X_{f}))-X_{f}(\eta(R_{t}))=-X_{f}(1)=0
\end{equation}
and for all $X\in\mathfrak{X}(M)$ we have
\begin{equation}
\begin{split}
(-[X_{f},R_{t}]\lrcorner d\theta)(X)&=d\theta([R_{t},X_{f}],X)\\
&=Xd\theta(R_{t},X_{f})-R_{t}d\theta(X,X_{f})+X_{f}d\theta(X,R_{t})\\
&-d\theta([X,R_{t}],X_{f})+d\theta([X,X_{f}],R_{t})\\
&=R_{t}d\theta(X_{f},X)+d\theta(X_{f},[X,R_{t}])\\
&=R_{t}(df(X)-(R_{t}f)\eta(X))+df([X,R_{t}])-(R_{t}f)\eta([X,R_{t}])\\
&=R_{t}(df(X))-(R_{t}R_{t}f)\eta(X)-(R_{t}f)R_{t}(\eta(X))+X(df(R_{t}))\\
&-R_{t}(df(X))-(R_{t}f)(X(\eta(R_{t}))-R_{t}(\eta(X)))\\
&=-(R_{t}R_{t}f)\eta(X)+X(R_{t}f)\\
&=d(R_{t}f)(X)-(R_{t}R_{t}f)\eta(X),
\end{split}
\end{equation}
therefore $-[X_{f},R_{t}]\lrcorner d\theta=d(R_{t}f)-(R_{t}R_{t}f)\eta$, so that $X_{R_{t}f}=-[X_{f},R_{t}]$.
\begin{equation}
[E_{H},X_{f}]=[X_{H}+R_{t},X_{f}]=[X_{H},X_{f}]+[R_{t},X_{f}],
\end{equation}
so that
\begin{equation}
\begin{split}
[E_{H},X_{f}]&=X_{\lbrace f,H\rbrace}+X_{R_{t}f}\\
&= X_{\lbrace f,H\rbrace}+R_{t}f\\
&=0.
\end{split}
\end{equation}
We also have that $R_{z}$ is a symmetry of $E_{H}$ since $[R_{z},E_{H}]=[R_{z},X_{H}]+[R_{z},R_{t}]=X_{\lbrace H,1\rbrace}+X_{R_{t}(1)}=0$ and $R_{z},X_{f_{1}},\cdots,X_{f_{n}}$ generate a solvable Lie algebra with the Lie bracket of vector fields. In addition we have that the vector fields $R_{z},X_{f_{1}},\ldots,X_{f_{n}}$ are tangent to the submanifold $M_{f}$ because $X_{f_{i}}f_{j}=\lbrace f_{j},f_{i}\rbrace$ which is zero on $M_{f}$ and $R_{z}f_{i}=0$; also the evolution vector field $E_{H}$ is tangent to the submanifold $M_{f}$ since $E_{H}f_{i}=0$ for $i=1,\ldots,n$. Then the solutions of the Hamilton equations that live on $M_{f}$ are solutions of the vector differential equation $\dot{x}=E_{H}(t,x)$ with $(t,x^{1},\ldots,x^{n+1})$ local coordinates on $M_{f}$ with $x^{n+1}=z$.  Since the vector fields $R_{z},X_{f_{1}},\ldots,X_{f_{n}}$ generate a solvable Lie algebra of symmetries of the vector field $E_{H}$, then the system defined by the equation $\dot{x}=E_{H}(t,x)$ is solvable by quadratures (theorem \ref{theorem1}), that is, the solutions of the Hamilton equations that live in $M_{f}$ can be found by quadratures.
\rule{5pt}{5pt}

It is worth remarking that in this section we have restricted ourselves to good contact and cocontact Hamiltonian systems; it would be interesting to study dissipative Hamiltonian systems, for example the damped harmonic oscillator (with and without external forces) is represented as a (nongood) contact Hamiltonian system \cite{BG2021,GLR2022,LL2019,BCT2017,Letal2022}. For general contact Hamiltonian systems the concept of constants of motion (conserved quantity) is generalized to the concept of dissipated quantities, namely functions dissipated at the same rate as the Hamiltonian function \cite{LL2020,BG2021,GGMRR2020}. The study of the implications of the existence of a solvable Lie algebra of dissipated quantities of a contact or cocontact Hamiltonian system is considered for future works.

\section{Conclusions}

We have shown that a nonabelian notion of integrability by quadratures is possible even in the contact and cocontact frameworks. The existence of a solvable Lie algebra of constants of motion for a symplectic, cosymplectic, contact or cocontact Hamiltonian system allows us to find solutions of the equations of motion by quadratures. Of course we have the Liouville theorem as a corollary because every Abelian Lie algebra is trivially a solvable Lie algebra.

\section*{Acknowledgments}

The author wishes to thank CONACYT (México) for the financial support through a postdoctoral fellowship in the program Estancias Posdoctorales por México 2022. In addition the author thanks Dr. Adrian Escobar for his comments.

\bibliography{refs} 
\bibliographystyle{unsrt} 

\end{document}